\documentclass[letterpaper,12pt]{article}

\usepackage{amsmath,amssymb,amsthm,bbm,tabularx}
\usepackage{natbib}
\usepackage{hyperref}
\hypersetup{colorlinks=false,pdfborder={0 0 0}}
\usepackage{color}
\usepackage{here}
\usepackage{graphicx}
\usepackage[small]{caption}
\usepackage[margin=1in]{geometry}

\newcommand\B{\mathcal B}
\newcommand\E{\mathbb E}
\newcommand\I{\mathbbm1}
\renewcommand\L{\mathcal L}
\renewcommand\P{\mathbb P}
\renewcommand\S{\mathcal S}

\newcommand\cites[1]{\citeauthor{#1}'s (\citeyear{#1})}
\newcommand\MR[1]{\href{http://www.ams.org/mathscinet-getitem?mr=#1}{MR#1}.}

\newtheorem{theorem}{Theorem}
\newtheorem{lemma}[theorem]{Lemma}

\newtheorem{proposition}[theorem]{Proposition}
\theoremstyle{definition}
\newtheorem{definition}[theorem]{Definition}
\newtheorem{remark}[theorem]{Remark}

\newcommand\algorithm[2]{
\noindent
\begin{tabularx}{\textwidth}{|X|}
\hline
\vspace{-2pt}
\textbf{Algorithm #1.} #2\\
\hline
\end{tabularx}
}

\begin{document}

\def\spacingset#1{\renewcommand{\baselinestretch}%
{#1}\small\normalsize} \spacingset{1}

\begin{center}
\huge Expand and Contract\\[.1cm]
\large Sampling graphs with given degrees and other combinatorial families\\[.1cm]
\normalsize James Y. Zhao\\[.5cm]

\parbox{.85\textwidth}{
\small
\textit{Abstract.} Sampling from combinatorial families can be difficult. However, complicated families can often be embedded within larger, simpler ones, for which easy sampling algorithms are known. We take advantage of such a relationship to describe a sampling algorithm for the smaller family, via a Markov chain started at a random sample of the larger family. The utility of the method is demonstrated via several examples, with particular emphasis on sampling labelled graphs with given degree sequence, a well-studied problem for which existing algorithms leave much room for improvement. For graphs with given degrees, with maximum degree $O(m^{1/4})$ where $m$ is the number of edges, we obtain an asymptotically uniform sample in $O(m)$ steps, which substantially improves upon existing algorithms.
}
\end{center}

\section{Introduction}

Combinatorial structures are a natural tool for encoding the multitude of discrete data in our information-based world. When one wants to make a statistical inference on this data, it often becomes necessary to compare an observed structure to a typical one, which leads to the problem of sampling from the uniform measure on a family of combinatorial structures.

For the simplest families---for example, subsets of a given set, lattice points in a box, or graphs with a given vertex set---it is typically easy to sample uniformly. However, applications often demand additional constraints---for example, subsets with a given size, lattice points with a given norm, or graphs with a given degree sequence---resulting in more complicated families for which uniform sampling algorithms can be much more elusive.

Such a constraint creates a natural embedding of the complicated family (which we will call $\S_0$) within the simpler family (which we will call $\S$); from this, a sampling algorithm for $\S$ can be used to derive one for $\S_0$. The simplest such algorithm is to repeatedly sample from $\S$ and discard samples until a member of $\S_0$ is obtained; this is easy and exactly uniform, but often one must wait an exponentially long time.

Faced with exponentially many unwanted samples, a natural solution is to adjust them until they are in $\S_0$. We will introduce a framework where the adjustment is done by a Markov chain, which improves at each step with respect to a graded partition $\S=\S_0\sqcup\S_1\sqcup\cdots\sqcup\S_k$ of the larger family, until eventually $\S_0$ is reached. In the ideal case, expanding from $\S_0$ to $\S$ achieves uniformity by taking advantage of a uniform sampling algorithm for $\S$, and this uniformity is preserved by the Markov chain during contraction back to $\S_0$.

\begin{figure}[H]
\centering\parbox{.8\textwidth}{
\centering
\includegraphics[width=.6\textwidth]{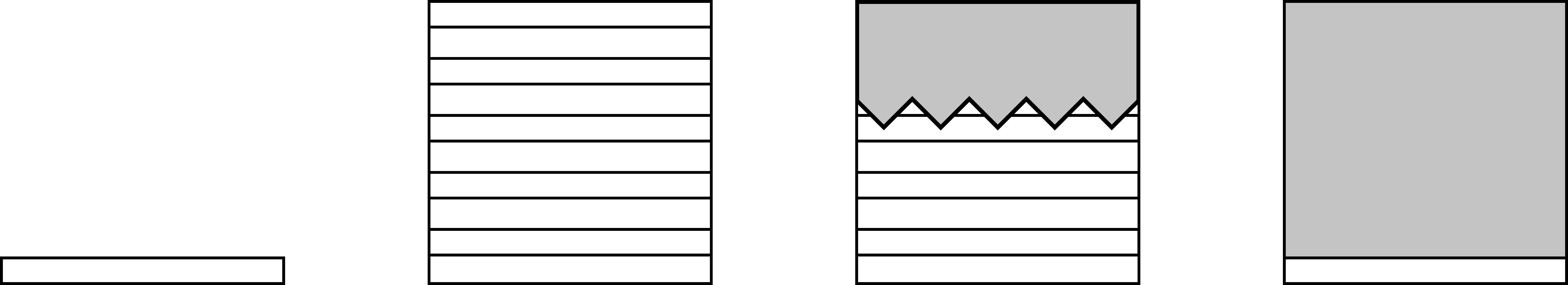}
\caption[Expand and Contract]{Expand and Contract. From left to right: (1) the set $\S_0$; (2) expansion to $\S=\S_0\sqcup\cdots\sqcup\S_k$; (3) and (4) contraction back to $\S_0$. Uniformity is gained through expansion; the goal is the ensure it is not lost during contraction.}}
\end{figure}

The idea of finding additional symmetry though expanding the state space is not a new one. It appears throughout the theory of generating functions, particularly the technique of Poissonisation \citep{wilf} and the more recent Boltzmann Sampler \citep{duchon}. In a more statistical setting, the techniques of auxilliary variables and data augmentation \citep{andersen} also rely on this idea. For the problem of sampling graphs with given degrees, \cite{mckaywormald} took a very similar approach to ours, generating a simple graph by first generating a multigraph, then ironing out any loops and multiple edges.

Our contribution is a general framework to apply this idea in a combinatorial setting, and detailed analyses of several examples. It is typically easy to compute the running time of our algorithm, and while it can be difficult to determine how close the output is to uniform, coupling techniques can be very successful here, since we will often have a Markov chain whose initial state and whose transitions are both uniform in some sense.

\subsection{General results}
\algorithm{A}{
Expand and Contract: sampling from a finite set $\S_0$.\\[4pt]
Inputs:\\[-18pt]
\begin{enumerate}
\parskip 0cm
\item A finite superset $\S\supset\S_0$, for which a sampling algorithm exists.
\item A graded partition $\S=\S_0\sqcup\S_1\sqcup\cdots\sqcup\S_k$. For $x\in\S_i$, define the \textbf{badness} $b(x)=i$.
\item An irreducible, symmetric Markov chain $Q$ on $\S$.
\end{enumerate}
\vspace{-.2cm}
Algorithm:\\[-18pt]
\begin{enumerate}
\parskip 0cm
\item[(1)] Let $X_0$ be a random point in $\S$ according to the existing sampling algorithm.
\item[(2)] Run the chain $Q$, rejecting moves that increase badness. More precisely, considering $Q$ as a probabilistic function $\S\rightarrow\S$, let
\vspace{-.3cm}
\[X_{t+1}=Q(X_t)\I_{b(Q(X_t))\le b(X_t)}+X_t\I_{b(Q(X_t))>b(X_t)}.\]
\vspace{-.9cm}
\item[(3)] Stop when a member of $\S_0$ is reached. More precisely, output $X_\tau$, where
\vspace{-.3cm}
\[\tau=\min\{t\ge0:b(X_t)=0\}.\]
\vspace{-1.2cm}
\end{enumerate}
}\vspace{4pt}

Let $Q^*$ denote the modified chain, that is, let $X_{t+1}=Q^*(X_t)$. Observing that Algorithm A is simply a specialisation of the Metropolis algorithm, we immediately have the following.

\begin{proposition}
\label{prop:metropolis}
The uniform measure on $\S_0$ is a stationary distribution of $Q^*$. Therefore, if $Q^*$ is irreducible, then the law of $X_t$ converges to uniform in total variation as $t\rightarrow\infty$.
\end{proposition}

There are a number of modifications that can be made to Algorithm A which can improve its performance in certain situations.

One issue is $Q^*$ may not be irreducible, and in particular, the chain might get stuck in an absorbing set disjoint from $\S_0$. This problem is alleviated by replacing the deterministic rejection of bad moves with a statistical mechanical model where moves are rejected probabilistically according to an energy function $H$ associated with each level of badness.

Formally, let $Q^*_H$ be the chain $Q$ where moves that changes the badness from $i$ to $j$ are accepted with probability $\min(e^{H(i)-H(j)},1)$. For example, if $H(i)=Ci$ for some constant $C\ge0$, then a decrease in badness is always accepted, while an increase in badness by $\Delta b$ is accepted with probability $\alpha^{\Delta b}$, where $\alpha=e^{-C}$ is interpreted as the rate of acceptance of bad moves. Note that $\alpha=1$ recovers the original chain $Q$ which accepts all moves, while $\alpha=0$ yields the chain $Q^*$ which always rejects bad moves.

This statistical mechanical setting carries a wealth of further enhancements. For example, Simulated Tempering \citep{marinari,madras} suggests starting with $\alpha=1$ and gradually decreasing $\alpha$. Often, the initial state $X_0$ is the stationary distribution of $Q$ ($\alpha=1$), while the limiting distribution $X_\infty$ is the stationary distribution of $Q^*$ ($\alpha=0$), so this may allow the chain to stay close to stationary at all times.

Another possibility is the strategy of \cite{wang}, which suggests dynamically adjusting the energy function to repel away from frequently observed badness levels. This is achieved by observing the badness $b$ at each step, and adjusting the energy function $H$ by increasing the value of $H(b)$. This favours moves which enter a level of badness that is not frequently observed, and can help prevent the chain from getting stuck by increasing the energy associated with a region in which the chain spends a lot of time.

One final important modification is to run the chain for additional steps after $\S_0$ is reached. Since Proposition \ref{prop:metropolis} gives convergence when $t\rightarrow\infty$, we can obtain any desired uniformity of output by simply running the chain for a large number of extra steps, and indeed, this is the idea behind Markov Chain Monte Carlo (MCMC) techniques. However, the goal here is to only prescribe a small number of additional steps, and in some cases, such as the example in Section \ref{sec:multistar}, the Expand and Contract part of the algorithm mixes sufficiently well that very few extra steps are needed to prove a bound on uniformity.

\vspace{8pt}
\algorithm{A$'$}{
Expand and Contract: bells and whistles version.\\[4pt]
Inputs:\\[-18pt]
\begin{enumerate}
\parskip 0cm
\item A finite superset $\S\supset\S_0$, for which a sampling algorithm exists.
\item A graded partition $\S=\S_0\sqcup\S_1\sqcup\cdots\sqcup\S_k$.
\item A irreducible, symmetric Markov chain $Q$ on $\S$.
\item (Optional) An energy function $H:\{1,\ldots,k\}\rightarrow\mathbb R$ on the badness.
\item (Optional) An dynamic modification to the energy function at each step.
\item (Optional) An prescribed number $T$ of additional steps to run.
\end{enumerate}
\vspace{-.2cm}
Algorithm:\\[-18pt]
\begin{enumerate}
\parskip 0cm
\item[(1)] Let $X_0$ be a random point in $\S$ according to the existing sampling algorithm.
\item[(2)] Run the chain $Q$. Reject moves that increase the badness, or with an energy function, accept moves from that change badness from $i$ to $j$ with probability $\min(e^{H(i)-H(j)},1)$. Dynamically modify the energy function if required.
\item[(3)] When a member of $\S_0$ is reached, run a further $T$ steps and stop.
\vspace{-10pt}
\end{enumerate}
}

\newpage
Again, convergence follows immediately from the Metropolis algorithm. Note that we are excluding dynamic adjustments to the energy function, since it can result in a non-Markovian process, which is beyond the scope of this paper.

{
\renewcommand{\thetheorem}{\ref{prop:metropolis}$'$}
\addtocounter{theorem}{-1}
\begin{proposition}
With a fixed energy function, as $T\rightarrow\infty$, the output of Algorithm A$'$ converges in total variation to a measure whose restriction to $\S_0$ is uniform.
\end{proposition}
}

However, for practical applications, our primary interest is in behaviour at finite times. In particular, we would like to know how long to run the algorithm for an output that is close to uniform. This will be demonstrated through a variety of motivating examples in Section \ref{sec:examples}, and our primary application, labelled graphs with given degree sequence, in Section \ref{sec:graph}.

Typically, it is fairly easy to bound the running time: bounding the probability of reducing the badness at each step yields a bound on the running time by a sum of independent geometric random variables with those probabilities of success. Uniformity is more difficult, but lends itself well to a coupling argument. The initial state $X_0$ contains all the symmetry provided by the sampling algorithm on $\S$, and will typically be close to uniform on $\S$. We can define a process $Y_t$, started at $Y_0=X_0$, whose restriction to $\S_0$ remains uniform at all times. Then, it suffices to show that $X_\tau$ remains coupled to $Y_\tau$ with high probability, and thus is close to uniform. This strategy of proof will be demonstrated in Sections \ref{sec:examples} and \ref{sec:graph}.

\subsection{Graphs with given degrees}

Section \ref{sec:graph} is devoted to the problem of sampling labelled graphs with given degree sequence, which is an important problem in the statistical study of graphs and networks, and for which existing algorithms are not always practical. See Section \ref{sec:graphintro} for a full review of the literature.

In this problem, $(d_1,\ldots,d_n)$ is a degree sequence (that is, a sequence of positive integers which prescribes the degree at each vertex of a labelled graph); $n$ denotes the number of vertices, $m=\frac12\sum_id_i$ denotes the number of edges, and we assume $d_1\le d_2\le\cdots\le d_n$. The goal is to sample uniformly from the set of graphs with this degree sequence. Most existing algorithms require a sparseness condition where $d_n$ is a small power of $n$ or $m$, while algorithms that do not require such a condition either produce an output of unproven uniformity or take unpractically long to run.

In Section \ref{sec:graphmain}, we will show that Algorithm A succeeds in at most $O(m\log m)$ steps under the assumption $d_n=o(m^{1/2})$, and the output is asymptotically uniform under the assumption $d_n=o(n^{1/4})$. By optimising the algorithm using the fact that we can know a priori which moves of the chain will have no effect on the distribution, we can reduce the running time to $O(d_n^2)$ under the same assumptions. Since it takes $O(m)$ steps to initialise the algorithm, the overall runtime is $O(m)$, which is substantially better than existing algorithms.

In Section \ref{sec:multistar}, we consider \textit{multi-star graphs}, a class of non-sparse degree sequences which is important in that it serves as a counterexample to almost every existing algorithm. We show that Algorithm A$'$, with $T=O(\log m)$ and no energy function, produces an asymptotically uniform output in $O(m^2)$ steps, which can again be optimised to $O(m)$. While we can only prove uniformity up to degree $d_n=o(m^{1/4})$ in the general case, the fact that the algorithm works well for this non-sparse example where most others fail suggests that the sparseness constraint may not be necessary.

\section{Motivating examples}
\label{sec:examples}
This section comprises simple examples that motivate the Expand and Contract idea, in preparation for the main application of graphs with given degrees in Section \ref{sec:graphmain}.

\subsection{Subsets of a given size}
Let $\S_0$ be the subsets of $[n]=\{1,2,\ldots,n\}$ with cardinality $k$, where $\epsilon n\le k\le(1-\epsilon)n$ as $n\rightarrow\infty$ for some $\epsilon>0$. The best existing algorithms for sampling from $\S_0$ \citep{knuth,pak} run in time $O(n)$. To use Expand and Contract, we need the following inputs:\\[-18pt]
\begin{enumerate}
\parskip 0cm
\item Superset: let $\S=2^{[n]}$ be all subsets of $[n]$, from which we can sample by including each element of $[n]$ independently with probability $k/n$. Note that this is uniform restricted to the subsets of a given cardinality, but is not uniform on all of $2^{[n]}$.
\item Partition: let $\S_i$ be the subsets of $[n]$ of cardinality $k\pm i$.
\item Markov chain: from a starting subset of $[n]$, pick an element of $[n]$ uniformly, and choose to include or exlude this element from the subset with probability $\frac12$ each.
\end{enumerate}

\noindent Then, the Expand and Contract algorithm is as follows:\\[-18pt]
\begin{enumerate}
\parskip 0cm
\item Pick a random subset of $[n]$ by including elements independently with probability $k/n$.
\item Pepeatedly pick an element of $[n]$ and include or exclude it with probability $\frac12$ each, rejecting moves which cause the cardinality to move further away from $k$.
\item When the cardinality reaches exactly $k$, stop and output the subset.
\end{enumerate}

\begin{proposition}
\label{prop:subsets}
This algorithm yields a uniformly random $k$-subset of $[n]$ in $O(n)$ steps.
\end{proposition}
\begin{proof}
Suppose the initial cardinality is $c$, and without loss of generality assume $c>k$. At each step with cardinality $i$, the cardinality decreases by 1 with probability $i/2n$, so the time to reach cardinality $k$ is bounded by a sum of geometric random variables with probability of success $i/2n$. Thus, the expected number of steps is
\[\frac{2n}{c}+\frac{2n}{c-1}+\cdots+\frac{2n}{k+1}=2n\big(\log c-\log k+O(\tfrac1k)\big).\]
Since $c=k+O(\sqrt n)$ and $k=\Theta(n)$, we have $c/k=1+O(n^{-1/2})$, hence the expected number of steps is $O(\sqrt n)$. Using Markov's inequality to bound the number of steps by its mean, we see that the runtime is dominated by the $O(n)$ steps to pick the initial subset. Finally, since both the initial state and the Markov chain are invariant under relabelling of points, the output is uniform by symmetry.
\end{proof}

\subsection{Permutations}
\label{sec:permutations}
Let $\S_0=S_n$ be the permutations of $n$. The best existing algorithms \citep{knuth,pak} are similar to the previous example, and run in $O(n)$ steps. By considering permutations as bijective functions $[n]\rightarrow[n]$, and expanding to the superset of all functions, we obtain the following Expand and Contract algorithm:\\[-18pt]
\begin{enumerate}
\parskip 0cm
\item Pick a uniformly random function $[n]\rightarrow[n]$.
\item Repeatedly pick a uniformly random element of $[n]$ and map it to another uniformly random element of $[n]$, rejecting moves which reduce the cardinality of the range.
\item Output the function when it is surjective.
\end{enumerate}
If the range has cardinality $n-i$, then the probability of increasing the cardinality is least $i^2/n^2$. As in the previous example, we can bound by a sum of geometric random variables to obtain a runtime of $O(n^2)$, with the output uniform by symmetry.

However, note that there are two layers of symmetry---in both the domain and the codomain---but we only need one for uniformity. Instead of picking points in the domain uniformly, we can maintain a list of ``bad points'' whose values already occur in the range, and always choose from these. The symmetry in the codomain is unaffected, so the output is still uniformly random, but the runtime improves substantially to $O(n\log n)$. Optimisations like this will be important for the graphs with given degrees application in Section \ref{sec:graphmain}.

\subsection{General linear group}
Let $\S_0=GL_n(\mathbb F_q)$ be the invertible $n\times n$ matrices over a finite field $\mathbb F_q$. The best existing algorithms \citep{shahshahani,pak} for sampling from $\S_0$ run in time $O(n^3)$. We can Expand and Contract via the superset of all matrices as follows:\\[-18pt]
\begin{enumerate}
\parskip 0cm
\item Pick $n$ random vectors in $\mathbb F_q^n$ as the columns of a $n\times n$ matrix.
\item Repeatedly pick a random column and replace it with a random vector in $\mathbb F_q^n$, rejecting moves that decrease the rank of the matrix.
\item Output the matrix when it reaches full rank.
\end{enumerate}
Similarly to the permutations example, there are two layers of symmetry---in both the rows and the columns---but only one is required. Instead of picking a random column, we can always choose a column which is a linear combination of other columns, which can be determined from the reduced row echelon form. Under this optimisation, both the initial state and the transition probabilities of the Markov chain depend only on the linear dependences between columns, which are invariant under Gaussian elimination. Since all invertible matrices are in the same orbit under Gaussian elimination, the output remains exactly uniform.

The initial rank is $n-O(1)$ with high probability, and we pick a vector outside the span of the existing columns with probability $\Theta(1)$, so the algorithm takes $O(1)$ steps. Since each step requires finding the reduced row echelon form, the overall runtime is $O(n^3)$.

\subsection{Lattice points in a sphere}
Let $\S_0$ be the set of tuples $(a_1,\ldots,a_n)$ of non-negative integers with $a_1^2+\cdots+a_n^2=E$. This has the quantum mechanical interpretation of configurations of $n$ non-interacting bosons in an infinite one-dimensional well, with a given energy $E$, which we consider to grow as $E=Cn$ for some constant $C$. Physical background for this model, as well as analytical results, can be found in the upcoming paper of \cite{chatterjee2}. We can Expand and Contract via the superset of all tuples of non-negative integers as follows:\\[-18pt]
\begin{enumerate}
\parskip 0cm
\item Pick each component of an $n$-tuple independently from a geometric distribution $G$ with probability of success chosen so that $\E[G^2]=C$. This results in an overall energy close to $E$, with the choice of geometric distribution corresponding to the physical fact that the energies of the individual particles should follow the Boltzmann distribution.
\item Repeatedly pick a component at random and either add or subtract 1 with equal probability, rejecting moves where the energy gets further away from $E$, or moves that introduce a negative component.
\item Output the tuple when the energy reaches $E$.
\end{enumerate}

Again, by bounding the probability of reducing the badness, the runtime can be shown to be $O(n)$. However, it is difficult to say anything about the uniformity of the output. Simulations suggest it is fairly close to uniform for small $n$, but the exponential growth of the state space makes it difficult to extend this observation to larger values. This example shows the typical case---runtime bounds are often quite easy to prove even for fairly complicated examples, while uniformity is usually difficult to determine.

\subsection{Magic squares}
Let $\S_0$ be the set ot $n\times n$ magic squares, that is, the set of $n\times n$ matrices whose entries are a permutation of $\{1,\ldots,n^2\}$ and whose row, column and diagonal sums are all equal to $\frac12n(n^2+1)$. Despite being a famous combinatorial family which has been studied since antiquity, the existing literature contains little discussion of uniform sampling. \cite{ardell} showed how to generate a magic square using a genetic algorithm, while \cite{kitajima} used a technique similar to ours to estimate the number of magic squares up to size $n\le30$, but neither result bounds the running time or the distance to uniformity. We can use Expand and Contract as follows:\\[-18pt]
\begin{enumerate}
\parskip 0cm
\item Pick a random $n\times n$ matrix whose entries are a permutation of $\{1,\ldots,n^2\}$.
\item Define the ``badness'' as the $L^1$ difference between the row, column and diagonal sums from the desired value of $\frac12n(n^2+1)$. Repeatedly swap two random entries\vspace{-.03cm}, rejecting moves which increase badness from $b$ to $c$ with probability $\alpha^{c^2-b^2}$. We choose a quadratic energy function since getting stuck is more likely when the badness is already small, and $\alpha=0.937$ was determined by minimising the average empirical runtime.
\item Stop and output when a magic square is obtained.
\end{enumerate}

We are unable to prove any results about either the running time or the uniformity of the algorithm. However, it does seem to work reasonably well in practice, demonstrating the robustness of the Expand and Contract strategy in being able to write down a plausible algorithm for almost any combinatorial family. For the 7040 magic squares of size $4\times4$, the empirical distribution was distance 0.130 from uniform in total variance, and the empirical probability of any magic square was within a factor of 2.196 from uniform. For larger squares, the extremely large state space prevented a similar simulation, but the algorithm nonetheless succeeded in producing magic squares up to size $50\times 50$, after which the runtime became prohibitive on an \textit{AMD Phenom II X4 765} processor.

\section{Graphs with given degree sequence}
\label{sec:graph}
\subsection{Introduction}
\label{sec:graphintro}
The problem of sampling labelled graphs with given degree sequence arises in the study of statistics on networks. For example, \cite{grossman} calculated that the mean Erd\H os number (among mathematicians who have a finite Erd\H os number) is 4.65. This leads to the following question: is this a surprising sociological consequence of mathematicians' publishing habits, or is it a purely graph theoretical phenomenon? In order to compare the observed real-world collaboration graph to ``typical'' ones, we want to sample labelled graphs (since mathematicians are distinguishable) with given degree sequence (since changing Erd\H os' propensity to collaborate would substantially alter the question).

Other interesting applications include Darwin's finches (Charles Darwin observed that the species of finches on the Gal\'apagos Islands tend to have differently shaped beaks when they coexist on the same island, suggesting that competition for the same food sources drove some species to extinction; see \cite{chen} for a statistical analysis of the bipartite graph connecting finches and islands), the six degrees of separation phenomenon (social networks tend to have surprisingly low distances between random people, but is this a sociological or a graph theoretical phenomenon?), and network motifs (the study of whether graphs have a statistically significant number of certain subgraphs, such as positive feedback loops in directed graphs of biochemical pathways).

The problem of sampling graphs with given degrees began as an afterthought to the study of their enumeration. A classical sampling algorithm via perfect matchings is credited to \cite{bender} and \cite{bollobas}, and produces a multigraph (where loops and multiple edges are allowed) with the given degree sequence.

For the remainder of this chapter, let $d=(d_1,\ldots,d_n)$ refer to a degree sequence, ascendingly ordered by $d_1\le d_2\le\cdots\le d_n$, and assumed to be valid in the sense that there exists at least one simple graph with that degree sequence. Also let $n$ denote the number of vertices, and $m=\frac12\sum_id_i$ denote the number of edges.

\vspace{10pt}
\algorithm{B}{
Classical sampling algorithm for multigraphs with given degrees.\\[-18pt]
\begin{enumerate}
\parskip 0cm
\item[(1)] Divide the first vertex into $d_1$ ``edge-endpoints'', the second vertex into $d_2$ edge-endpoints, and so on, until one obtains $\sum_id_i=2m$ edge-endpoints in $n$ groups.
\item[(2)] Choose a random matching of these $2m$ edge-endpoints, by picking edges between pairs of unmatched edge-endpoints at random until all are matched.
\item[(3)] Collapse the groups of edge-endpoints back into vertices to obtain a labelled multigraph with the given degree sequence.
\vspace{-10pt}
\end{enumerate}
}\vspace{4pt}

An easy calculation shows that if one repeats Algorithm B and waits until a simple graph is obtained, then the output is exactly uniform on simple graphs \citep{bayati}. Indeed, the earliest results on uniform sampling \citep{wormald} proved that if the degrees are bounded, then there is a bounded positive probability of obtaining such a simple graph, showing that it works reasonably well for bounded degree sequences.

\begin{figure}[H]
\centering
\parbox{.8\textwidth}{
\hfill
\includegraphics[width=2cm,angle=90]{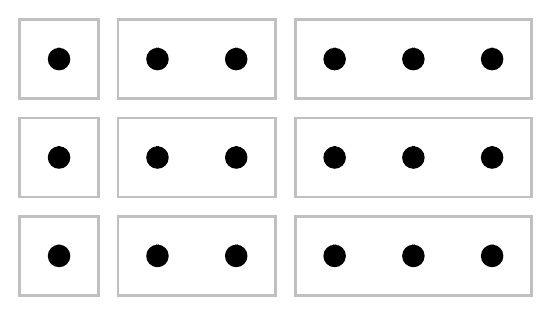}
\hfill\hfill
\includegraphics[width=2cm,angle=90]{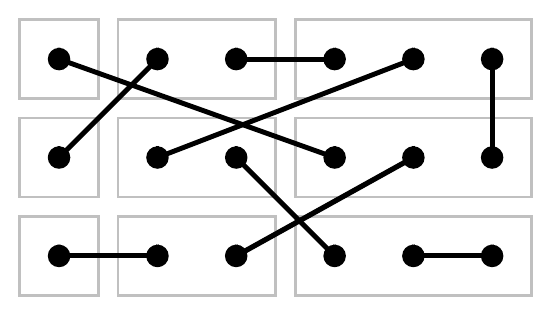}
\hfill\hfill
\includegraphics[width=2cm,angle=90]{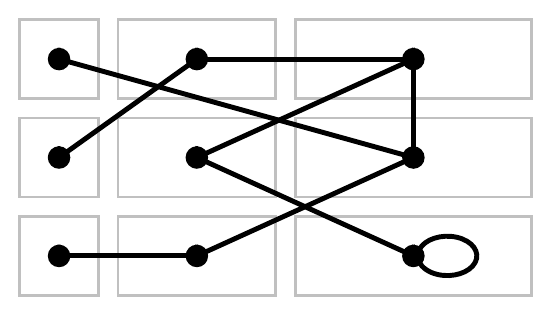}
\hfill\hfill

\caption[Sampling labelled multigraphs with given degree sequence]{
Example for Algorithm B: sampling labelled multigraphs with given degree sequence $(1,1,1,2,2,2,3,3,3)$. From left to right: (1) dividing the vertices into groups of edge-endpoints; (2) randomly matching pairs of edge-endpoints; and (3) collapsing groups of edge-endpoints back into vertices. Note that this realisation contains a loop, so the output is not a simple graph.
}
}
\end{figure}

Many subsequent papers attempted to modify Algorithm B to work for less sparse degree sequences. \cite{mckaywormald} showed that as long as the maximum degree is $d_n=O(m^{1/4})$, there is a way to remove loops and multiple edges to obtain a uniformly random simple graph, at the cost of additional computations to determine how these loops and multiple edges can be removed, which results in a relatively slow running time of $O(m^2d_n^2)$, which is quite costly in practice. \cite{steger} found an alternative modification---dynamically deciding to accept or reject the addition of edges with a certain probability---that allows a much faster running time of $O(nd_n^2)$, which works for $d_n=o(n^{1/11}(\log n)^{-3/11})$. \cite{bayati} improved upon this idea with a smarter rejection criterion, which was able to greatly expand the sparseness condition to $d_n=O(m^{1/4-\epsilon})$ for any $\epsilon>0$.

The algorithm of \cite{blitzstein} also follows this pattern of choosing the edges sequentially, but instead matches edges starting from the highest degree vertices in a way that ensures the algorithm always terminates. This avoids the possibility of getting stuck at the cost of not being able to determine the uniformity of the result, but the amount of non-uniformity can be explicitly computed, allowing for importance sampling techniques to yield unbiased estimates of statistical quantities. \cite{chen} took an alternative approach in the case of bipartite graphs by picking all the edges for a whole vertex at a time. \cite{blanchet} showed that for this algorithm, if the two vertex sets have maximum degree $o(m^{1/2})$ and $O(1)$, the output is asymptotically uniform with runtime $O(n^3)$, and again importance sampling can be used to adjust for non-uniformity in other cases. However, despite producing unbiased estimates, importance sampling does not necessarily work well when the distribution is exponentially far from uniform, so closeness to uniform is still an important question for these algorithms.

\cite{bezakova2} found an explicit counterexample demonstrating that the algorithm of \cite{chen} can be exponentially far from uniform. This counterexample is important, as it applies to all of the algorithms discussed thus far: all of them experience exponentially bad behaviour, in either runtime or uniformity or both. In Section \ref{sec:multistar}, we expand this counterexample to a larger class of degree sequences that nonetheless still causes exponentially bad behaviour---we call them \textit{multi-star graphs}---and show that our new strategy works much better in this case.

An alternate school of thought \citep{jerrum1,jerrum2,jerrum3,bezakova} is to apply Markov Chain Monte Carlo: start at an instance of a simple labelled graph with the given degree sequence, and adjust it slightly so that the degree sequence does not change, repeating until a new, almost-independent sample is reached. This is proven to run in polynomial time, but unfortunately the polynomial bound---$O(n^4m^3d_n(\log n)^5)$---is difficult to use in practice, particularly since Markov Chain Monte Carlo does not simply terminate, but instead requires us to prescribe a number of steps to run.

Our Expand and Contract algorithm is a hybrid of the classical perfect matching algorithm with MCMC techniques, and is most similar to the algorithm of \cite{mckaywormald} in that we start with a multigraph and obtain a simple graph by removing loops and multiple edges. The key difference is that we use a Markov chain to remove these bad edges without requiring any precomputation, which allows for both a more general framework and a greatly reduced running time of $O(m)$. This runtime can be broken down into $O(m)$ steps for Algorithm B plus $O(\sqrt m)$ steps to remove the bad edges while maintaining asymptotic uniformity, which is asymptotically best possible in the sense that Algorithm B can be implemented so that it comprises exactly $m$ edge-picking steps with $O(1)$ overhead, and those $m$ edge-picking steps are necessary to generate a graph with $m$ edges.

A comparison of algorithms follows in Figure \ref{fig:algorithms}.

\vspace{6pt}
\begin{figure}[H]
\centering
\parbox{.8\textwidth}{
\hfill\begin{tabular}{|l|l|l|}
\hline
Algorithm & Runtime & Sparseness\\
\hline
Perfect Matching & $O(m)$ & $d_n=O(1)$ \\
McKay-Wormald & $O(m^2d_n^2)$ & $d_n=O(m^{1/4})$ \\
Steger-Wormald & $O(nd_n^2)$ & $d_n=o(n^{1/11})$ \\
Bayati-Kim-Saberi & $O(md_n)$ & $d_n=o(m^{1/4})$ \\
Blitzstein-Diaconis & $O(n^2m)$ & Importance\footnotemark[1] \\
Chen-Diaconis-Holmes-Liu & $O(n^3)$ & Importance\footnotemark[1] \\
Jerrum-Sinclair & Polynomial\footnotemark[2] & $d_n\le\sqrt{n/2}$ \\
Bez\'akov\'a-Bhatnagar-Vigoda & $O(n^4m^3d_n)$ & All\footnotemark[3] \\
Expand and Contract & $O(m)$ & $d_n=o(m^{1/4})$ \\
\hline
\end{tabular}\hfill\hfill

\caption[Comparison of algorithms for graphs with given degree sequence]{\label{fig:algorithms}
\parindent 1em
Comparison of sampling algorithms for graphs with given degree sequence. Runtime and sparsity constraints are listed with sub-polynomial terms omitted for succinctness. Expand and Contract is the fastest algorithm, and applies to the equal largest range of inputs among non-MCMC algorithms.

\vspace{8pt}
Notes: 1. Blitzstein-Diaconis and Chen-Diaconis-Holmes-Liu rely on importance sampling to overcome known non-uniformity of the output;~ 2. Jerrum-Sinclair showed their algorithm runs in polynomial time without explicitly computing it;~ 3. Bez\'akov\'a-Bhatnagar-Vigoda is asymptotically uniform for all degree sequences.
}
}

\end{figure}

\subsection{Expand and Contract}
\label{sec:graphmain}
Since Algorithm B provides an easy way to generate a random labelled multigraph with the given degree sequence, this problem fits naturally into our Expand and Contract strategy. The number of \textbf{bad edges}---loops and multiple edges, counting multiplicity---defines a natural graded partition of our expanded state space. Thus, the only remaining ingredient we need to use Algorithm A is a Markov chain on multigraphs. However, before we introduce the Markov chain, we will first prove the following bound on the initial number of bad edges, which is a minor refinement of the bound of \cite{mckaywormald}.

\begin{lemma}
\label{lem:badedges}
If $d_n=o(\sqrt m)$, then the number of bad edges in the output of Algorithm B is $b=O(d_n^2)$ with high probability, in the sense that there is $C<\infty$ such that $\P[b>Cd_n^2]\rightarrow0$.
\end{lemma}
\begin{proof}
We can write
\[b=\sum_i\sum_{k\ge 1}E_{i,i,k}+\sum_{i<j}\sum_{k\ge2}E_{i,j,k},\]
where $E_{i,j,k}$ is the indicator function of the multiplicity of edge $(i,j)$ being at least $k$. By Chebychev's inequality, it suffices to show that $\E[b]=O(d_n^2)$ and ${\mathrm{Var}}(b)=o(d_n^4)$.

Since we will be making asymptotic estimates of many quantities, for convenience, let $A\simeq B$ denote $A=(1+o(1))B$ as $n\rightarrow\infty$, and similarly let $A\lesssim B$ denote $A\le(1+o(1))B$ as $n\rightarrow\infty$.

Observe that the total number of possible perfect matchings between $2m$ edge-endpoints is $(2m!)/2^mm!$. Thus, the number of matchings containing any specified set of $k=o(\sqrt m)$ matches is
\begin{equation}
\label{eqn:prob}
\left.\frac{(2m-2k)!}{2^{m-k}(m-k)!}\middle/\frac{(2m)!}{2^mm!}\right.=\frac{2^km(m-1)\cdots(m-k+1)}{2m(2m-1)\cdots(2m-2k+1)}\simeq(2m)^{-k}.
\end{equation}
There are ${d_i\choose 2k}\frac{(2k)!}{2^kk!}$ matchings that create at least a $k$-tuple loop at $i$, hence
\begin{equation}
\label{eqn:loopbound}
\P[E_{i,i,k}=1]\simeq{d_i\choose 2k}\frac{(2k)!}{2^kk!}(2m)^{-k}\le\frac{d_i^{2k}}{(4m)^kk!}.
\end{equation}
Then, for $d_n=o(\sqrt m)$,
\[\E\left[\sum_i\sum_{k\ge 1}E_{i,i,k}\right]\lesssim\sum_{k\ge1}\frac{d_n^{2k-1}\sum_id_i}{(4m)^kk!}=\frac m{d_n}\left(\exp\left(\frac{d_n^2}{4m}\right)-1\right)\simeq\frac{d_n}4.\]
Similarly, there are ${d_i\choose k}{d_j\choose k}k!$ matchings that create at least a $k$-tuple edge at $(i,j)$ with $i\ne j$, hence
\begin{equation}
\label{eqn:medgebound}
\P[E_{i,j,k}=1]\simeq{d_i\choose k}{d_j\choose k}k!(2m)^{-k}\le\frac{d_i^kd_j^k}{(2m)^kk!}.
\end{equation}
Then, for $d_n=o(\sqrt m)$,
\[\E\left[\sum_{i<j}\sum_{k\ge2}E_{i,j,k}\right]\lesssim\frac12\sum_{k\ge2}\frac{\left(d_n^{k-1}\sum_id_i\right)^2}{(2m)^kk!}=\frac12\frac{m^2}{d_n^2}\left(\exp\left(\frac{d_n^2}{2m}\right)-1-\frac{d_n^2}{2m}\right)\simeq\frac{d_n^2}{16}.\]
Thus, $\E[b]=O(d_n^2)$, and it only remains to bound the variance. Write $b=\sum_\ell E_\ell$, where the sum in $\ell$ runs across the union of $L=\{(i,i,k):k\ge 1\}$ and $M=\{(i,j,k):i<j,k\ge2\}$, so that $\E[b^2]=\sum_{\ell,\ell'}\E[E_\ell E_{\ell'}]$.

We can split this sum depending on whether $\ell$ and $\ell'$ are in $L$ or $M$, and further split depending on whether the indices are equal. For $\ell\in L$ and $\ell'\in L$, if the indices are unequal, the same calculation in (\ref{eqn:loopbound}) yields that bound squared, while if they are equal, (\ref{eqn:loopbound}) holds but with $\max(k,k')$ instead:
\begin{align*}
\sum_{\ell\in L}\sum_{\ell'\in L}\E[E_\ell E_{\ell'}]
&\lesssim\sum_{i\ne i'}\sum_{k\ge1,k'\ge1}\frac{d_i^{2k}}{(4m)^kk!}\frac{d_{i'}^{2k'}}{(4m)^{k'}k'!}
+\sum_i\sum_{k\ge1,k'\ge 1}\frac{d_i^{2\max(k,k')}}{(4m)^{\max(k,k')}\max(k,k')!}\\
&\lesssim\left(\frac{d_n}4\right)^2+\sum_{k\ge1}k\frac{d_n^{2k-1}m}{(4m)^kk!}\simeq\frac{d_n^2}{16}+\frac{d_n}4.
\end{align*}
For $\ell\in L$ and $\ell'\in M$, we obtain the same bound of (\ref{eqn:loopbound}) times (\ref{eqn:medgebound}) regardless of whether $i=i'$ or $i=j'$:
\begin{align*}
\sum_{\ell\in L}\sum_{\ell'\in M}\E[E_\ell E_{\ell'}]
&\lesssim\frac12\sum_{i,i',j'}\sum_{k\ge1,k'\ge2}\frac{d_i^{2k}}{(4m)^kk!}\frac{d_{i'}^{k'}d_{j'}^{k'}}{(2m)^{k'}k'!}\lesssim\frac{d_n}4\frac{d_n^2}{16}=\frac{d_n^3}{64}.
\end{align*}
For $\ell\in M$ and $\ell'\in M$, we obtain the bound (\ref{eqn:medgebound}) squared unless $i=i'$ and $j=j'$, in which case we obtain (\ref{eqn:medgebound}) with $\max(k,k')$ instead:
\begin{align*}
\sum_{\ell\in M}\sum_{\ell'\in M}\E[E_\ell E_{\ell'}]
&\lesssim\frac14\sum_{i,j,i',j'}\sum_{k\ge2,k'\ge2}\frac{d_i^kd_j^k}{(4m)^kk!}\frac{d_{i'}^{k'}d_{j'}^{k'}}{(2m)^{k'}k'!}+{}\\
&\hspace*{1cm}\frac14\sum_{i,j}\sum_{k\ge2,k'\ge2}\frac{d_i^{\max(k,k')}d_j^{\max(k,k')}}{(4m)^{\max(k,k')}\max(k,k')!}\\
&\lesssim\E[b]^2+\frac14\sum_{k\ge2}k\frac{d_n^{2k-2}m^2}{(4m)^kk!}\simeq\E[b]^2+\frac{d_n^2}{64}.
\end{align*}
Hence, ${\mathrm{Var}}(b)=O(d_n^3)$, which shows concentration around the mean.
\end{proof}

There are many possible Markov chains that leave the degree sequence invariant. The simplest possibility is the \textbf{2-swap}: randomly remove two edges $(a,b)$ and $(c,d)$, and replace them with $(b,c)$ and $(d,a)$. This simple chain is sufficient for the special case in Section \ref{sec:multistar}, but not for the general case (see Remark \ref{rem:3swap}), where we will need the second simplest possibility, the \textbf{3-swap}: randomly remove three edges $(a,b)$, $(c,d)$ and $(e,f)$, and replace them with $(b,c)$, $(d,e)$ and $(f,a)$.

\begin{figure}[H]
\centering
\parbox{.8\textwidth}{
\hfill
\includegraphics[width=5cm]{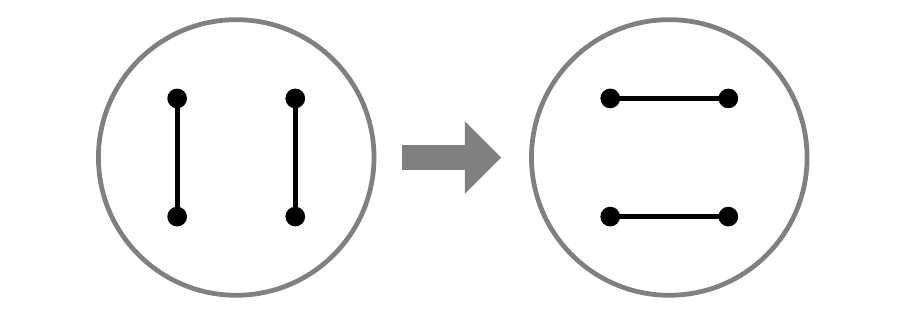}
\hfill
\includegraphics[width=5cm]{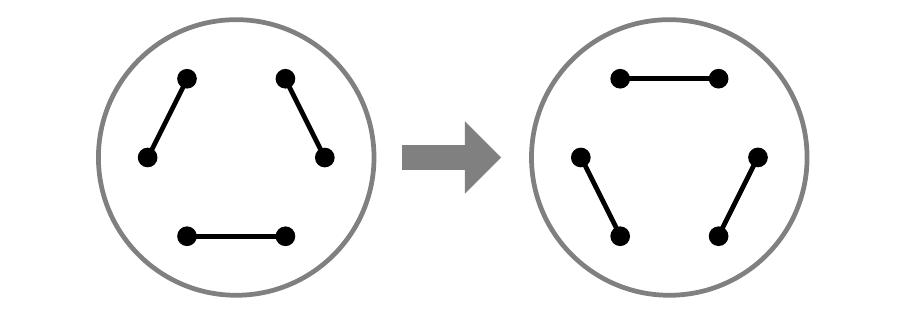}
\hfill\hfill

\caption[Markov chains for Algorithm A]{
The two simplest possible choices of Markov chain for use with Algorithm A: the 2-swap (left) and the 3-swap (right). In both cases, a cycle of vertex pairs consisting of alternating edges and non-edges has its edges and non-edges swapped.
}
}
\end{figure}

These moves were introduced by \cite{mckaywormald}, whose algorithm used 3-swaps to remove loops, and double 2-swaps to remove double edges. Our improvement is in using a simple Markov chain to decide which move to make, so as to avoid the quadratic-time precomputations that are necessary in the McKay-Wormald algorithm. An alternative chain is that of \cite{jerrum1,jerrum2}, which expands the state space to near-perfect matchings by allowing up to two edge-endpoints to remain unmatched.

\begin{figure}[H]
\centering
\parbox{.8\textwidth}{
\hfill
\includegraphics[width=5cm]{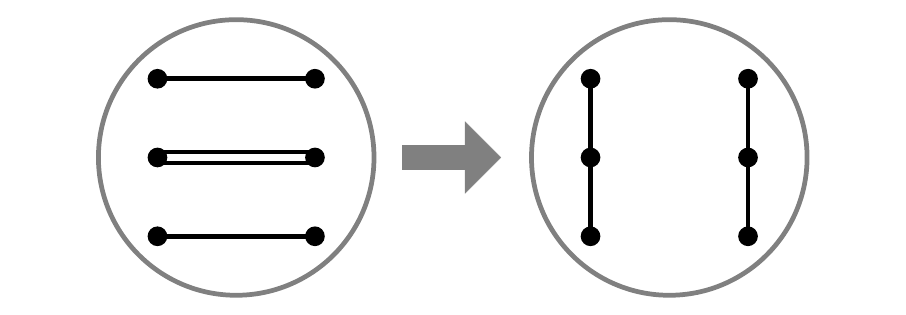}
\hfill
\includegraphics[width=5cm]{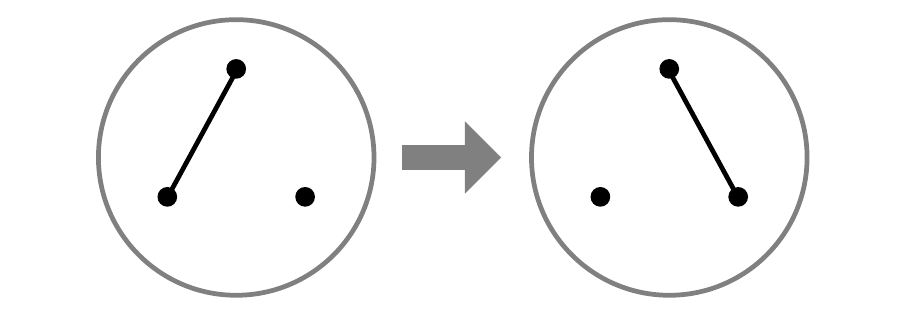}
\hfill\hfill

\caption[More markov chains for Algorithm A]{
Markov chains used in the literature for Algorithm A: \cites{mckay} double 2-swap (left), which is only used for double edges, and \cites{jerrum1} near-perfect matching chain (right), which allows up to two edge-endpoints to remain unmatched so that single edges can be moved around for faster mixing. Both of these chains can be used for Expand and Contract, but since we are able to obtain good results using the much simpler 3-swap, we will leave the discussion of which chain to use for future work.
}
}
\end{figure}

\algorithm{A}{Expand and Contract for sampling labelled graphs with given degree sequence $(d_1,\ldots,d_n)$.\\[4pt]
Inputs:\\[-18pt]
\begin{enumerate}
\parskip 0cm
\item Superset: let $\S$ be the set of multigraphs with the given degree sequence.
\item Grading: let $\S_i$ be the subset of multigraphs with $i$ bad edges.
\item Markov chain: random 3-swaps.
\end{enumerate}
\vspace{-.2cm}
Algorithm:\\[-18pt]
\begin{enumerate}
\parskip 0cm
\item[(1)] Run Algorithm B to obtain a random multigraph.
\item[(2)] Pick three edges $\{(a,b),(c,d),(e,f)\}$ randomly (also pick the vertex order randomly), and replace them with $\{(b,c),(b,e),(f,a)\}$, rejecting moves that increase the number of bad edges.
\item[(3)] Repeat until a simple graph is obtained.
\vspace{-10pt}
\end{enumerate}
}

\begin{theorem}
\label{thm:sqrt}
For graphs with given degree sequence with $n$ vertices, $m$ edges and maximum degree $d_n=o(\sqrt m)$, Algorithm A terminates in at most $(1+o(1))\frac23m\log d_n$ steps with high probability.
\end{theorem}
\begin{proof}
Let $(a,b)$ be a bad edge, either a loop or one edge in a multiple edge. If we can find two edges $(c,d)$ and $(e,f)$ such that $(b,c)$, $(d,e)$ and $(f,a)$ do not exist as edges, then a 3-swap will remove the bad edge $(a,b)$ without creating any new bad edges, thereby reducing the badness by 1.

If $(b,c)$ and $(c,d)$ both exist as edges, then $c$ must be one of the $o(\sqrt m)$ vertices adjacent to $b$, and $d$ must be one of the $o(\sqrt m)$ vertices adjacent to $c$, hence there are at most $o(m)$ choices for an edge $(c,d)$ such that $(b,c)$ is also an edge. Similarly, there are $o(m)$ edges $(e,f)$ such that either $(d,e)$ or $(f,a)$ is an edge. Hence, the number of choices of edges $(c,d)$ and $(e,f)$ so that $(b,c)$, $(d,e)$ and $(f,a)$ are not edges is $(2m-o(m))^2=(4-o(1))m^2$. (Note that the factor of 2 arises since there is also a choice of ordering for the vertex pairs.)

For each set of 3 edges, there are 8 possible 3-swaps: an edge $(a,b)$ can be followed by either $(c,d)$, $(d,c)$, $(e,f)$ or $(f,e)$, and in each case there are 2 choices for the order of the remaining vertex pair, each of which produces a distinct 3-swap. Thus, the total number of choices of 3-swaps is $8{m\choose 3}=\frac43m(m-1)(m-2)$, hence the probability of a 3-swap that removes the bad edge $(a,b)$ without creating any new bad edges is $(1+o(1))3/m$.

If there are $b$ bad edges, then the probability of reducing the badness is $(1+o(1))3b/m$. The total number of steps to remove all bad edges is then bounded by a sum of independent geometric random variables with probabilities of success given by $3b/m,3(b-1)/m,\ldots,6/m,3/m$ to within a factor of $1+o(1)$. By the law of large numbers, this sum is $(1+o(1))\frac m3\log b$ with high probability. Lemma \ref{lem:badedges} implies that $\log b\le(1+o(1))\log(d_n^2)$, which completes the proof.
\end{proof}

\begin{remark}
Theorem \ref{thm:sqrt} also holds for 2-swaps by the same argument, with the number of steps given by $(1+o(1))m\log d_n$. Note that the overall running time is the same: losing the factor of $\frac23$ is exactly cancelled by the fact that each step requires only 2 operations rather than 3. Similar results also hold for $k$-swaps for any fixed $k\ge2$.
\end{remark}

It remains to estimate the distance to uniformity. Consider $\P[X_0=x]$, the probability that Algorithm B will output a given multigraph $x\in\S$. This probability is proportional to the number of matchings that give rise to $x$, which is $\prod_id_i!$ for a simple graph since each permutation of the edge-endpoints at each vertex results in a distinct matching, and decreases by a factor of $2^kk!$ for each $k$-tuple loop and $k!$ for each $k$-tuple non-loop edge. In particular, the initial state is uniform restricted to each set of multigraphs with a given set of bad edges.

The idea is to couple the Markov chain $X_t$ initially with another chain that always stays uniform on each set of multigraphs with a given set of bad edges, and show that the probability of decoupling vanishes as $n\rightarrow\infty$, which implies that $X_\tau$ is asymptotically uniform on the set of simple graphs. The central argument is that with high probability, each move of the chain either does not change any bad edges or removes a single bad edge, and for such moves, the measure stays approximately uniform with an error whose total contribution vanishes.

However, this presents a problem: if a set of bad edges $B$ contains a double edge $(i,j,2)$, then reducing that double edge to a single edge cannot possible yield an approximately uniform distribution on the multigraphs with bad edge set $B\setminus\{(i,j,2)\}$, because the multiplicity of $(i,j)$ must be 1, while it is 0 for a typical multigraph satisfying our sparseness condition. On the other hand, once the multiplicity becomes 0, then the measure can be approximately uniform, because only a vanishing proportion of multigraphs have multiplicty 1 for any given pair of vertices. This suggests that we can solve this problem by also counting the single edge left over from removing a double edge as a bad edge, and requiring such bad edges to also be removed.

Thus, we will expand the definition of ``bad edge'' to also include any edge left over from reducing the multiplicity of a multiple edge. Note that this means that badness is no longer a function of the multigraph itself, but of the sample path of multigraphs since the initial state. However, this will not present an issue in the proof.

\begin{theorem}
\label{thm:uniform}
If $d_n=o(m^{1/4})$, then Algorithm A (using the modified definition of badness) can be optimised to terminate in time $O(m)$ with high probability, and the output is asymptotically uniform in total variation distance as $m\rightarrow\infty$.
\end{theorem}
\begin{proof}
We need to introduce notation to describe multigraphs with a common set of bad edges. Let $B=\{(i_1,j_1,k_1),\ldots,(i_\ell,j_\ell,k_\ell)\}$ be a set of vertex pairs $(i_1,j_1),\ldots,(i_\ell,j_\ell)$ with prescribed multiplicities $k_1,\ldots,k_\ell$. We will think of $B$ as a set of bad edges (under the modified definition). Let $S_B$ be the set of multigraphs containing $B$ with no other bad edges, that is, multigraphs whose multiplicities between vertex pairs in $B$ agree with $B$, and with no other loops or multiple edges.

Under this notation, our observation that the initial state is uniform on each set of multigraphs that share a set of bad edges becomes $\L(X_0\,|\,X_0\in S_B)=U_B$, where $\L(X|A)$ denotes the conditional distribution of the random variable $X$ restricted to the event $A$, and $U_B$ denotes the uniform measure on $S_B$.

Let $B_0,B_1,\ldots$ be such sets of bad edges, and let $\B_t=\{X_0\in S_{B_0},\ldots,X_t\in S_{B_t}\}$ be the event that $X_0,\ldots,X_t$ have bad edge sets $B_0,\ldots,B_t$ respectively. Then, let
\[\L(X_t\,|\,\B_t)=c_tU_{B_t}+(1-c_t)E_t,\]
where $0\le c_t\le1$ is the largest number so that the above equation holds with $E_t$ a signed measure with total absolute mass at most 1. This gives us the total variation bound $||\L(X_t|\B_t)-U_{B_t}||_{TV}\le1-c_t$, so the goal is to show that for almost every sequence of bad edge sets $B_0,B_1,\ldots$, we have $c_\tau\rightarrow1$ as $n\rightarrow\infty$. Note that given $\B_\infty=\bigcap_t\B_t$, the stopping time $\tau=\inf\{t\ge0:b(B_t)=0\}$ is deterministic, so there are none of the usual convergence issues associated with a random stopping time.

Call the sequence $B_0,B_1,\ldots$ \textbf{admissible} if consecutive bad edge sets satisfy either:
\begin{itemize}
\item $B_{t+1}$ is the same as $B_t$;
\item $B_{t+1}$ is $B_t$ with one non-simple edge is removed, that is, with one element $(i,j,\ell)$ replaced by $(i,j,\ell-1)$ where $\ell\ge2$; or
\item $B_{t+1}$ is $B_t$ with one simple edge removed $(i,j,1)$ removed.
\end{itemize}

For an admissible sequence of bad edge sets, we can estimate $c_t$ as follows.
\begin{itemize}
\item If $B_{t+1}=B_t$, then since the transition probabilities within $S_{B_t}$ are symmetric, uniformity is preserved, and hence $c_{t+1}=c_t$.

\item Suppose $B_{t+1}$ removes a loop from $B_t$, that is, replaces $(i,i,\ell)$ with $(i,i,\ell-1)$ where $\ell\ge2$. We will show that $c_{t+1}=(1-O(m^{-1/2}))c_t$, by showing that the law of $Q(U_{B_t})$ restricted to $\S_{B_{t+1}}$ is within a factor of $1-O(m^{-1/2})$ of uniform.

For any multigraph $x\in\S_{B_t}$, there are $\frac43m(m-1)(m-2)$ possible 3-swap transitions, and by the argument in Theorem \ref{thm:sqrt}, $(2m-O(d_n^2))^2=(1+o(m^{-1/2}))4m^2$ of them remove the loop $(i,i)$ without creating any new bad edges. Thus, for any multigraph $y\in\S_{B_{t+1}}$ which is adjacent to $A(y)$ points in $\S_{B_t}$ (counting multiplicity), the mass at $y$ under $Q(U_{B_t})$ is
\begin{equation}
\label{eqn:probratio}
\sum_xU_{B_t}(x)Q(x,y)=
A(y)\frac{1}{|\S_{B_t}|}\frac{(1+o(m^{-1/2}))4m^2}{\frac43m(m-1)(m-2)}.
\end{equation}
We can compute $A(y)$ by counting the number of 3-swaps that return $y$ to $\S_{B_t}$. This is given by the choice of a pair of non-loop non-coincident edges $(i,j)$ and $(i,k)$ (that is, with $i$, $j$ and $k$ distinct vertices), followed by the choice of a third edge $(j',k')$ such that $(j,j')$ and $(k,k')$ are not edges.

The number of choices for the pair of edges incident at $i$ depends only on $B_{t+1}$ and $i$, not on $y$ itself. For any given choice, again by the argument in Theorem \ref{thm:sqrt}, there are $O(d_n^2)$ choices that have an edge at $(j,j')$ and $O(d_n^2)$ choices that have an edge at $(k,k')$, and hence there are $2m-O(d_n^2)=(1-o(m^{-1/2}))2m$ choices that have neither edge.

This shows that $A(y)$ fluctuates by a factor of $o(m^{-1/2})$, hence the mass at $y\in\S_{B_{t+1}}$ under $Q(U_{B_t})$ also fluctuates by a factor of $o(m^{-1/2})$, from which we can conclude that $c_{t+1}=(1-O(m^{-1/2}))c_t$.

\item If $B_{t+1}$ removes a non-simple non-loop edge from $B_t$, replacing $(i,j,\ell)$ with $(i,j,\ell-1)$ where $\ell\ge2$, a similar counting argument shows that again, we have $c_{t+1}=(1-O(m^{-1/2}))c_t$.

\item In both the above cases, if $\ell=1$, the same argument shows that the resulting distribution is within $1-O(m^{-1/2})$ of uniform on the subset of $S_{B_{t+1}}$ consisting of multigraphs which do not contain the edge $(i,j)$. However, (\ref{eqn:prob}) shows that the edge $(i,j)$ occurs with probability at most $d_id_j/2m=o(m^{-1/2})$, so again, we have $c_{t+1}=(1-O(m^{-1/2}))c_t$.
\end{itemize}

Since the number changes in $B_t$ is at most the initial (modified) badness, which is $o(\sqrt m)$ by Lemma \ref{lem:badedges}, and each change results in a factor of $1-O(m^{-1/2})$ change to $c_t$, it follows that $c_\tau=(1-O(m^{-1/2}))^{o(\sqrt m)}=1-o(1)$.

Thus, given any admissible sequence of bad edge sets, $X_\tau$ is asymptotically uniform on simple graphs. By Theorem \ref{thm:sqrt}, $B_t$ is eventually constant, so there are only \mbox{countably} many sequences of bad edge sets (except for a set of sequences with probability 0). Since the bound on uniformity depends only on the degree sequence and not on the sequence of bad edge sets, it follows that $X_\tau$ is asymptotically uniform on simple graphs across all admissible sequences. Hence, it only remains to show that almost every sequence is admissible.

A sample path can lie outside an admissible sequence only if it undergoes an inadmissible transition, where more than one bad edge is modified in one step. If there are $b$ bad edges, this can occur by picking a bad edge along with either another bad edge or an edge adjacent to the bad edge, which occurs with probability at most $2b(b+2d_n)/m^2=O(b^2/m^2)$. Since the chain has $b$ bad edges for $O(m/b)$ steps, the total probability of this occuring is at most
\[\sum_bO(b^2/m^2)O(m/b)=O(b^2/m)=o(1).\]
Finally, note that only the steps which remove one bad edge only are crucial in our argument. Thus, the algorithm can be optimised to always choose a 3-swap between one bad edge and two good edges at each step, with rejection if any new bad edges are created, which can be achieved in $O(1)$ operations by maintaining both an array of edges with a block of bad edges followed by a block of good edges, and a hash table of multiplicities between vertex pairs.

By Lemma \ref{lem:badedges}, there are $o(\sqrt n)$ bad edges, so the run time is dominated by the $O(m)$ steps required for the initial multigraph generated via Algorithm B.
\end{proof}

\begin{remark}
\label{rem:3swap}
It is necessary to use 3-swaps rather than 2-swaps in order for (\ref{eqn:probratio}) to concentrate within a factor of $o(m^{-1/2})$. Using 2-swaps, the quantity $A(y)$ would be enumerated by pairs of non-loop non-coincident edges $(i,j)$ and $(i,k)$ such that $(j,k)$ does not exist as an edge, for which the fluctuations are $o(1)$ rather than $o(m^{-1/2})$, which is not enough to prove Theorem \ref{thm:uniform}. However, strengthening the sparseness condition to $d_n=o(m^{1/8})$ results in the fluctuations becoming small enough for Theorem \ref{thm:uniform} to hold with 2-swaps.
\end{remark}

\begin{remark}
The modified definition of badness is not only useful for the proof, but is also much easier to implement in practice, since it guarantees that a bad edge can only become good through its own removal, not through the removal of any other edge. There is very little cost to using this modified definition: while the number of steps the chain needs to run is roughly doubled, the number of operations needed to keep track of the bad edges at each step is also roughly halved. Nonetheless, as a theoretical point of interest, we conjecture that the modification is not actually necessary, since the probability that one of the $o(\sqrt m)$ multiple edges occurs between any given vertex pair is $O(1/m)$, so we would expect the effects of leftover edges to cancel out over the entire probability space. 
\end{remark}

\subsection{Multi-star graphs}
\label{sec:multistar}
\begin{definition}
A \textbf{multi-star} degree sequence is one of the form $(1,\ldots,1,d_1,\ldots,d_k)$, where the number of 1s is equal to $d_1+\cdots+d_k$. This corresponds to the degree sequence of a disjoint union of $k$ star graphs. We consider $k$ to be fixed, while the degrees may grow. Call the vertices of degree 1 \textbf{leaves}, and call the vertices with degrees $d_1,\ldots,d_k$ \textbf{hubs}.
\end{definition}

Note that for $k=1$, this corresponds to a star graph, and for $k=2$, this corresponds to the counterexample of \cite{bezakova2}, who showed that for such degree sequences, the algorithm of \cite{chen} produces a bipartite graph that is exponentially far from uniform. In fact, every existing non-MCMC sampling algorithm suffers the same exponentially bad behaviour under this single counterexample. On the other hand, while MCMC has polynomial runtime for this example, the provable bounds are not practically useful. Thus, this is a very important case in the problem of sampling graphs with given degree sequence, for which the Expand and Contract strategy works extremely well, as we will show in this section.

In the $k=2$ case, there are only two equivalence classes under relabelling---the two hubs are either adjacent or not adjacent; the remaining edges are determined by this choice, up to relabelling. A quick computation yields $2c\choose c-1,c-1,2$ graphs with adjacent hubs and $2c\choose c$ graphs with non-adjacent hubs, so the probability that the hubs are not adjacent is around $2/c^2$ under the uniform measure. The direct sampling algorithms that choose one edge at a time have a $\Theta(1)$ chance to choose the edge between the hubs at each of $\Theta(c)$ steps, and thus yield a graph where the probability that the hubs are non-adjacent is $\exp(-\Theta(c))$. This is exponentially far from uniform under the following definition.

\begin{figure}[H]
\centering
\parbox{.8\textwidth}{
\hfill
\includegraphics[width=4cm]{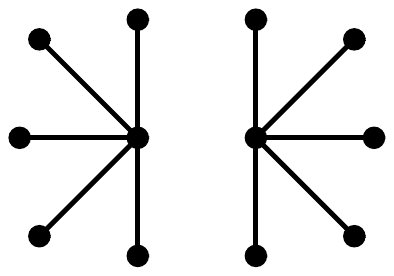}
\hfill
\includegraphics[width=4cm]{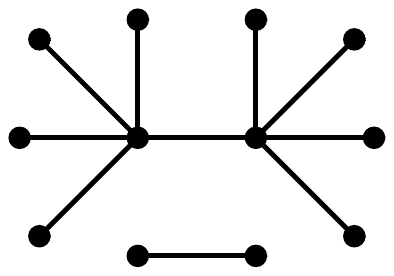}
\hfill\hfill

\caption[Multi-star graphs with $k=2$]{
The two equivalence classes of multi-star graphs when $k=2$, drawn for $c=5$. The class on the left (with the hubs non-adjacent) is $c^2/2$ times larger than the class on the right (with the hubs adjacent).
}
}
\end{figure}

\begin{definition}
The \textbf{probability ratio metric} between two probability measures $\mu$ and $\nu$ on a common finite set $S$ is
\[d(\mu,\nu)=\max_{x\in S}\big|\log\mu(x)-\log\nu(x)\big|.\]
We say these two measures are exponentially far apart when $d(\mu,\nu)$ grows super-linearly with $\log n$, which corresponds to the existence of a state $x\in S$ such that the ratio $\mu(x)/\nu(x)$ gets exponentially large or exponentially small. In particular, we are insisting that the probability of every state is close to that of the target measure, which is a significantly stronger condition than more commonly used distances such as total variation.
\end{definition}

The need for a stronger metric for uniformity comes from the fact that in the $k=2$ case, the ratio between the sizes of equivalence classes is $2/c^2$, which converges to 0, so one of them (the one where the two hubs are not adjacent) has vanishingly small probability. Under total variation distance, an algorithm which never produces a graph in this equivalence class could nonetheless be asymptotically uniform, which is not a desirable trait. This is avoided by using the probability ratio metric.

The algorithms of \cite{mckaywormald}, \cite{steger}, \cite{bayati} and \cite{blitzstein} all produce outputs that are exponentially far from uniform in this example, by the argument above. The MCMC algorithms may not perform exponentially badly, but do not fare much better---the algorithm of \cite{jerrum2} is only proven to run in polynomial time for maximum degree $d_n\le\sqrt{n/2}$, while the algorithm of \cite{bezakova} does provably run in polynomial time but is particularly unwieldy, carrying the prohibitively expensive run time of $O(n^8(\log n)^5)$, as well as requiring that one starts at a graph generated by a fairly complicated bootstrapping process.

On the other hand, we will show that Expand and Contract performs remarkably well. Note that the degee sequence does not satisfy the conditions of Theorem \ref{thm:sqrt}; in some sense it is as far as possible from satisfying the conditions. The fact that we nonetheless obtain good results suggest that the strategy is far more robust than we are currently able to prove.

We will return to the old definition of badness which does not include simple edges left over from removing a multiple edge, as well as using 2-swaps instead of 3-swaps. This is purely for the reason of making the analysis and discussion simpler; the same results carry through regardless of the choice. We will also need a prescribed number of additional steps $T$ for which the chain should be run after reaching a simple graph.

\vspace{8pt}
\algorithm{A$'$}{
Expand and Contract for multi-star graphs.
\begin{enumerate}
\item[(1)] Run Algorithm B to obtain a random multigraph.
\item[(2)] Pick two edges $\{(a,b),(c,d)\}$ randomly (also pick vertex order randomly) and replace them with $\{(b,c),(d,a)\}$, rejecting moves that increase the number of bad edges.
\item[(3)] When a simple graph is reached, run a further $T$ steps and stop.
\vspace{-10pt}
\end{enumerate}
}\vspace{4pt}

\begin{theorem}
\label{thm:multistar}
For any multi-star degree sequence, there exists $T=O(\log m)$ such that Algorithm A$'$ can be optimised to terminate in time $O(m)$ with high probability, and the resulting graph is within $O(1/m)$ of uniform in the probability ratio metric.
\end{theorem}
\begin{proof}
A bad edge can only occur between two hubs, and since the total hub degree is the same as the total leaf degee, for each bad edge, there must be a disconnected edge between two leaves. Thus, if there are $b$ bad edges, there are at least $2b^2$ choices of edge flip (counting orientation) that decrease the badness, hence the probability of decreasing the badness is at least $2b^2/m^2$. The number of steps to remove all bad edges is bounded by a sum of geometric random variables with this probability of success, which is $O(m^2)$ with high probability. This can be optimised to $O(m)$, which will be discussed towards the end of the proof.

Let $\tilde\S$ and $\tilde\S_0$ denote the equivalence classes of $\S$ and $\S_0$ respectively under relabelling of leaves. Both the initial distribution and the moves of the Markov chain $X_t$ on $\S$ are symmetric with respect to the labelling, so $\tilde X_t$, the equivalence class of $X_t$, is a Markov chain on $\tilde\S$ with absorbing set $\tilde\S_0$.

Let $\pi$ be the (uniform) stationary distribution of $X_t$, and let $\tilde\pi$ be the stationary distribution of $\tilde X_t$. For any $x\in\S_0$ with equivalence class $\tilde x\in\tilde\S_0$, we have the relation
\[\frac{\tilde\pi(\tilde x)}{\pi(x)}=|\tilde x|=\frac{\P(\tilde X_t=\tilde x)}{\P(X_t=x)}.\]
Hence, the distance to uniform (under the probability ratio metric) in $\S_0$ is exactly equal to the distance to uniform in $\tilde\S_0$, so it suffices to prove that the $\epsilon$-mixing time (again, under the probability ratio metric) of the $\tilde X_t$ chain is $O(\log m)$.

The remainder of the proof relies on the idea that equivalence classes are determined by the edges between the hubs. This shows that the size of $\tilde\S_0$ is bounded; we also find that the probability of adding an edge is $\Theta(1)$, while the probability of moving or removing an edge is $\Theta(1/m^2)$. Roughly speaking, we will show that a Markov chain whose state space is bounded in size and which has these transition probabilities will mix in time $O(\log m)$.

For notational convenience, let $K={k\choose2}$, and drop the tildes: we will be referring to the equivalence classes for the remainder of the proof, so $x$ will mean $\tilde x$, and so on. Let $e(x)$ denote the number of pairs of hubs in $x$ which are non-adjacent, and observe that $0\le e(x)\le K$. Let $x_0$ be the equivalence class corresponding to a complete graph between hubs, with $e(x_0)=0$.

\begin{lemma}
The stationary distribution satisfies $\pi(x)=\Theta(m^{-2e(x)})$, the probability of adding an edge between hubs is $\Theta(1)$, the probability of removing an edge betwen hubs is $\Theta(1/m^2)$, and the probability of moving an edge from one pair of hubs to another is $\Theta(1/m)$.
\label{lem:probs}
\end{lemma}

\begin{proof}
An equivalence class $x$ whose hubs $1,\ldots,k$ are connected to $c_1,\ldots,c_k$ other hubs has cardinality
\[\frac{(d_1+\cdots+d_k)!}{(d_1-c_1)!(d_2-c_2)!\cdots(d_k-c_k)!\,2^{(c_1+\cdots+c_k)/2}((c_1+\cdots+c_k)/2)!}.\]
Removing an edge $(i,j)$ reduces the size of the equivalence class by a factor of
\[(d_i-c_i+1)(d_j-c_j+1)/(c_1+\cdots+c_k)=\Theta(m^2).\]
The stationary probabilities of equivalence classes are proportional to their cardinality. Thus, $x_0$ has the highest probability, equivalence classes with $e(x)=1$ have a smaller probability by a factor of $\Theta(m^2)$, and so on. Since the number of states is bounded, this shows that $\pi(x)=\Theta(m^{-2e(x)})$.

There are three types of transitions in the chain $X_t$:
\begin{enumerate}
\item Adding an edge $(i,j)$ between hubs can occur when the edge flip is chosen between a hub-leaf edge at $i$ and another hub-leaf edge at $j$, which occurs with probability $(d_i-c_i)(d_j-c_j)/m(m-1)=\Theta(1)$.
\item Removing an edge $(i,j)$ between hubs can occur when the edge flip is chosen between the hub-hub edge $(i,j)$ and any leaf-leaf edge, which occurs with probability $(c_1+\cdots+c_k)/m(m-1)=\Theta(1/m^2)$. Note that there are $(c_1+\cdots+c_k)/2$ leaf-leaf edges, but we also multiply by 2 because both orientations of the edge flip will remove the edge $(i,j)$.
\item Moving an edge between hubs from $(i,j)$ to $(i,k)$ can occur when the edge flip is chosen between $(i,j)$ and a hub-leaf edge at $k$, which occurs with probability $(d_k-c_k)/m(m-1)=\Theta(1/m)$. \qedhere
\end{enumerate}
\end{proof}

\begin{lemma}
\label{lem:inductionstart}
There exists $T_0=O(\log m)$ such that for all $t\ge T_0$ and equivalence classes $x$,
\[\P(X_{\tau+t}=x)=\Theta(m^{-2e(x)}).\]
\end{lemma}

\begin{proof}
By Lemma \ref{lem:probs}, at each step, $e(X_t)$ decreases by 1 with probability $\Theta(1)$ unless it is already 0, and increases by 1 with probability $\Theta(1/m^2)$ unless it is already $K$. Let $e_t$ be a Markov chain on $\{0,1,\ldots,K\}$, starting at $K$, which decreases by 1 with probability $\delta$ and increases by 1 with probability $1/\delta m^2$, except at the endpoints. For sufficiently small $\delta$, we can couple $X_t$ and $e_t$ so that $e(X_{\tau+t})\le e_t$ for all $t$.

For the chain $e_t$, the probability of being at 0 within the next $K$ steps is bounded below by $\delta^K$, thus the probability that the sample path has not passed through 0 in $t$ steps is at most $(1-\delta^K)^{\lfloor t/K\rfloor}$. Hence, for all $t\ge T'_0=2K^2\log m/|\log(1-\delta^K)|$, the sample path has not passed through 0 with probability at most $m^{-2K}$.

On the other hand, if the sample path has passed through 0, then the probability of being in state $e$ is at most $(\delta m)^{-2e}$. We can prove this by induction: it is clearly true if the chain is currently at 0, and if it is true at any step, then it remains true at the next step, since any state $e$ will have probability at most
\[\left(1-\delta-\frac1{\delta m^2}\right)(\delta m)^{-2e}+\frac1{\delta m^2}(\delta m)^{-2(e-1)}+\delta(\delta m)^{-2(e+1)}=(\delta m)^{-2e}.\]
Thus, for any $t\ge T'_0=O(\log m)$,
\[\P(e(X_{\tau+t})\ge e)\le\P(e_t\ge e)=O(m^{-2e}).\]
This proves the upper bound $\P(X_{\tau+t}=x)=O(m^{-2e(x)})$. For the lower bound, observe that the upper bound implies $\P(X_{\tau+t}=x_0)=\Theta(1)$. Since edges are removed with probability $\Theta(1/m^2)$ at each step, at the next step, the probability of every state with $e(x)=1$ is at least $\Theta(1/m^2)$, and at the next step, the probability of every state with $e(x)=2$ is at least $\Theta(1/m^4)$, and so on. This yields the lower bound for $T_0=T'_0+K=O(\log m)$.
\end{proof}

Lemmas \ref{lem:probs} and \ref{lem:inductionstart} imply that the probability of every state other than $x_0$ is $O(1/m^2)$, under both the stationary distribution and $X_{\tau+t}$ for $t\ge T_0$, hence $\P(X_{\tau+t}=x_0)$ is within a factor of $1+O(1/m^2)$ from stationary. We will use induction to show that after $O(\log m)$ further steps, the probability of every other state is also within a factor of $1+o(1)$ from stationary.

\begin{lemma}
\label{lem:induction}
Suppose there exists $T_\ell=O(\log m)$ such that for all $t\ge T_\ell$ and all equivalence classes $x$ with $e(x)\le\ell$, $\P(X_{\tau+t}=x)=(1+O(1/m))\pi(x)$. Then there exists $T_{\ell+1}=O(\log m)$ satisfying the same property with $\ell+1$ instead of $\ell$.
\end{lemma}

\begin{proof}
Consider a state $x$ with $e(x)=\ell+1$, and suppose $t>\max(T_0,T_\ell)$. Any state $x'$ adjacent to $x$ has $e(x')=\ell$, $\ell+1$ or $\ell+2$; we can determine their probabilities and their contribution to the total mass entering and leaving $x$ using Lemmas \ref{lem:probs} and \ref{lem:inductionstart}.

For $e(x')=\ell+2$, the mass at $x'$ is $\Theta(m^{-2\ell-4})$, and the probability of entering $x$ is $\Theta(1)$, so the contribution to mass entering and leaving $x$ is $O(m^{-2\ell-4})$. For $e(x')=\ell+1$, the mass at $x'$ is $\Theta(m^{-2\ell-2})$, of which a proportion of $\Theta(1/m)$ enters $x'$, so the contribution is $O(m^{-2\ell-3})$. Since $\pi(x)=O(m^{-2\ell-2})$, these contributions are $O(1/m)\pi(x)$.

Now consider a state $x'$ with $e(x')=\ell$, which differs from $x$ by the addition of an edge $(i,j)$. Let $c_1,\ldots,c_k$ be the degrees of the hubs of $x$. By Lemma \ref{lem:probs}, $\pi(x')/\pi(x)=(d_i-c_i)(d_j-c_j)/2(K-\ell)$, and the probability of a transition to $x$ is $2(K-\ell)/m(m-1)$. By assumption, the probability at $x'$ is within $1+O(1/m)$ of stationary, so the mass entering $x$ from $x'$ is
\[\left(1+O\left(\frac1m\right)\right)\frac{(d_i-c_i)(d_j-c_j)}{m(m-1)}\pi(x).\]
Since such states $x'$ are enumerated by edges $(i,j)$ that can be added to $x$, the total mass entering $x$ from states $x'$ with $e(x')=\ell$ is
\[\sum_{(i,j)\notin x}\left(1+O\left(\frac1m\right)\right)\frac{(d_i-c_i)(d_j-c_j)}{m(m-1)}\pi(x)=\left(C_x+O\left(\frac1m\right)\right)\pi(x),\]
where $C_x$ is a constant depending only on $x$.

Mass exits from $x$ to $x'$ when two hub-leaf edges at $i$ and $j$ are chosen, which occurs with probability $(d_i-c_i)(d_j-c_j)/m(m-1)$. Again, this can occur for any pair $i$ and $j$ which is not an edge, hence the probability $p_t(x)=\P(X_{\tau+t}=x)$ satisfies the recursion
\[p_{t+1}(x)=\left(C_x+O\left(\frac1m\right)\right)\pi(x)+\left(1-C_x+O\left(\frac1m\right)\right)p_t(x).\]
Solving the recursion, the number of steps for $p_t$ to reach $(1+O(1/m))\pi(x)$ is $T_{\ell+1}(x)=O(\log\pi(x))=O(\log m)$. Since the state space of $x$ is bounded in size, we can take $T_{\ell+1}=\max_xT_{\ell+1}(x)=O(\log m)$, which completes the proof.
\end{proof}

By induction, Lemma \ref{lem:induction} show that there is $T=O(\log m)$ such $X_{\tau+T}$ is within a factor of $1+O(1/m)$ from stationary for all equivalence classes, which completes the proof of uniformity.

Finally, as with Theorem \ref{thm:uniform}, only the steps that actually remove a bad edge are important, so we can optimise the algorithm to always remove a bad edge at each step, which improves the quadratic running time to $O(m)$ without changing the bound on uniformity.
\end{proof}

\subsection{Erd\H os numbers}
As mentioned in the introduction to Section \ref{sec:graph}, the statistical study of Erd\H os numbers is a natural application for sampling graphs with given degree sequence. The underlying data structure is the collaboration graph, which is formed by taking the vertex set to be the set of all mathematicians, and adding an edge between every pair of mathematicians who have coauthored a publication. Many statistical questions can be answered by comparing the observed collaboration graph to typical graphs with the same degree sequence, where the rationale for fixing the degree sequence is that Erd\H os numbers would not be as interesting if Erd\H os did not have by far the highest number of collaborators of any mathematician.

Using the Expand and Contract algorithm described in Section \ref{sec:graphmain}, and the degree sequence of the collaboration graph generously provided by the Erd\H os Number Project \citep{grossman}, we generated random graphs with this degree sequence to study the mean Erd\H os number (the mean graph distance to Erd\H os among mathematicians within the same connected component). In 10,000 trials, the mean Erd\H os number of a random graph with the same degree sequence was 4.119 on average, with a sample standard deviation of 0.025, indicating very tight concentration around the mean.

In contrast, the actual mean Erd\H os number in the real world is 4.686\footnote{This mean Erd\H os number of 4.686 differs from the 4.65 cited in the introduction because it was computed using older data (2000), since the full collaboration graph for the newer data (2004) was not available.}, which is 22 standard deviations higher than the expected number. This indicates that there are strong sociological factors that govern the distribution of Erd\H os numbers, although it is not clear what these factors might one. One possible explanation is that mathematicians are more likely to find new collaborators who are already close to themselves in the collaboration graph, and thus the reduction in mean Erd\H os number from each new collaboration is lower than if collaborations were chosen at random.

\section{Conclusion}
For graphs with given degree sequence, the $O(m^{1/4})$ sparseness bound arises independently in three different algorithms: \cite{mckaywormald}, \cite{bayati} and now Expand and Contract. Thus, the most significant question remaining is whether fast uniform sampling can be extended to denser degree sequences. For regular graphs, the algorithm of \cite{bayati} works up to $d_n=O(m^{1/3-\epsilon})$, so a natural question is whether the same can be proven for Expand and Contract in the regular case.

Expand and Contract also runs extremely fast---the runtime of $O(m)$ is best possible in the sense that $\Omega(m)$ steps are necessary to generate the $m$ edges of the graph. An interesting question is whether linear-time sampling is possible for all degree sequences, and if not, whether we can find a lower bound on the runtime. The question of runtime also leads to the practical question of which sampling algorithm to choose in a given scenario. Most results in the literature are asymptotic in nature, but the extraction of explicit constants from the proofs is important for deciding which algorithm to use in real-world applications.

Additionally, the uniform sparseness bound is not particularly useful in the real world, where graphs tend to be mostly sparse with a few high degree vertices. Expand and Contract shows great promise here in that it provably works well for multi-star graphs, which can be considered the most extreme case of such a graph. Since the model is highly probabilistic in nature, one possible improvement is to change the bound on maximum degree to a bound on some notion of average degree.

Graphs with other restrictions---for example, requiring connectedness, given eigenvalues, or network motifs---are also important, as are countless other combinatorial families, such as permutations or matrices with various restrictions. Expand and Contract is a general strategy that can be applied to all of these, and the examples of Section \ref{sec:examples} show that it is highly versatile. We hope that many other useful algorithms will arise from this idea.

\bibsep 0pt

\end{document}